\documentclass[notitlepage, 12pt]{report}
\usepackage{graphicx}
\usepackage{amsmath, amssymb, dsfont}
\usepackage{cite}
\usepackage{titlesec, blindtext, color}
\definecolor{gray75}{gray}{0.75}
\newcommand{\hsp}{\hspace{20pt}}
\titleformat{\chapter}[hang]{\Huge\bfseries}{\thechapter\hsp\textcolor{gray75}{$|$}\hsp}{0pt}{\Huge\bfseries}

\newcommand{\qed}{\hfill \ensuremath{\Box}}

\newenvironment{proof}[1][Proof]{\begin{trivlist}
\item[\hskip \labelsep {\bfseries #1}]}{\end{trivlist}}

\newtheorem{theorem}{Theorem}[chapter]
\newtheorem{lemma}[theorem]{Lemma}
\newtheorem{proposition}[theorem]{Proposition}
\newtheorem{corollary}[theorem]{Corollary}

\newtheorem{definition}{Definition}[chapter]

\PassOptionsToPackage{pdftex,hyperfootnotes=false,pdfpagelabels}{hyperref}
\usepackage{hyperref}  
\hypersetup{
colorlinks=true, linktocpage=true, pdfstartpage=3, pdfstartview=FitV,
breaklinks=true, pdfpagemode=UseNone, pageanchor=true, pdfpagemode=UseOutlines,
plainpages=false, bookmarksnumbered, bookmarksopen=true, bookmarksopenlevel=1,
hypertexnames=true, pdfhighlight=/O, urlcolor=black, linkcolor=black, citecolor=black,
}   

\usepackage{ifthen} 
\newboolean{enable-backrefs} 
\setboolean{enable-backrefs}{false} 

\newcommand{\backrefnotcitedstring}{\relax} 
\newcommand{\backrefcitedsinglestring}[1]{(Cited on page~#1.)}
\newcommand{\backrefcitedmultistring}[1]{(Cited on pages~#1.)}
\ifthenelse{\boolean{enable-backrefs}} 
{
\PassOptionsToPackage{hyperpageref}{backref}
\usepackage{backref} 
\renewcommand*{\backref}[1]{}  
\renewcommand*{\backrefalt}[4]{
\ifcase #1 
\backrefnotcitedstring
\or
\backrefcitedsinglestring{#2}
\else
\backrefcitedmultistring{#2}
\fi}
}{\relax} 

\begin{document}

\title{{Classical Encryption and Authentication under Quantum Attacks}}
\author{Maria Velema}
\date{July 12, 2013}

\maketitle
\begin{abstract}
Post-quantum cryptography studies the security of classical, i.e.\ non-quantum cryptographic protocols against quantum attacks. 
Until recently, the considered adversaries were assumed to use quantum computers and behave like classical adversaries otherwise.
A more conservative approach is to assume that also the communication between the honest parties and the adversary is (partly) quantum.
We discuss several options to define secure encryption and authentication against these stronger adversaries who can carry out \emph{superposition attacks}. We re-prove a recent result of Boneh and Zhandry, stating that 
a uniformly random function (and hence also a quantum-secure pseudorandom function) can serve as a  message-authentication code which is secure, even if the adversary 
can evaluate this function in superposition.
\end{abstract}
\pagenumbering{roman}
\section*{Acknowledgements}
I would like to thank my supervisor Christian Schaffner for working together after making me interested with his course on cryptography.
I am grateful for his patience and didactic guidance through the complex proofs during my project. His help and confidence made this result possible.\\
I want to thank Ronald de Wolf for noticing the relation between the theorem of Boneh and Zhandry and Farhi's result. Apart from playing this crucial role, he made some very important comments. 
I also want to thank the other members of the thesis committee, Alexandru Baltag and Serge Fehr, for taking the time to read my thesis and for correcting some details.\\
I thank both Gerrit and Frank for their useful (and funny) comments from a different point of view.
Finally, I want to thank Nicole and Peter-Paul for studying together at the kitchen table.

\tableofcontents

\newpage
\pagenumbering{arabic}

\chapter{Introduction}
\label{intro}
To be prepared for a future in which \emph{quantum computers} have significant computing power but are not available to all users of cryptographic protocols, we need to study the security of classical (i.e.\ non-quantum) schemes against \emph{quantum attacks}. 
Usually in the field of \emph{post-quantum cryptography}, it is assumed that the honest parties are fully classical and the attackers or adversaries use quantum computers. This scenario can be generalized by dropping the assumption that the honest parties can compute and communicate only classically.  An important argument to do this is the fact that the world, as physical theories describe it nowadays, is not classical but quantum; although some devices are designed to work classically, they may leak quantum information under certain conditions. 
It is also conceivable that honest parties do use quantum computing, to speed up some computations, but still use classical schemes instead of quantum cryptography only. Since classical computing is a special form of quantum computing, considering quantum implementations of classical schemes captures a broader class of scenarios.

Notions of (classical) security 
are defined by means of a game between a \emph{challenger} and an \emph{adversary}. Different options to model the abilities of the adversaries lead to various security definitions. 
Damg{\aa}rd et al.~\cite{DFBS11} consider `superposition attacks' on some cryptographic protocols where  communication between the honest players and the adversaries can be quantum. The protocols they discuss are secret sharing, zero-knowledge protocols and multiparty computation.
Boneh and Zhandry~\cite{BZ13b} define quantum security games that model the case in which (part of) the communication is quantum, for both private-key and public-key encryption and authentication. 
They build upon earlier work in which 
pseudorandom functions are queried in superposition~\cite{Zha12b}.
Their research extends the ideas of the paper on quantum accessible random oracles~\cite{BDFLSZ11}.
As Damg{\aa}rd et al.\ 
point out, there is an important difference between superposition attacks on functions that are implemented by the adversary herself and superposition attacks on protocols 
run by the honest party. It seems to be similar, because both attacks are modelled in terms of oracles that may be queried in superposition. In the second case however, the adversary tricks the honest party to act like an oracle and some communication takes place, whereas in the first case the oracle is a subroutine of the algorithm that is modelled as black box. 

It is not always clear how (unintended) quantum communication between the adversary and the honest parties should be modelled. In particular, is it 
plausible that the adversary is able to create quantum entanglement with the sender? Superposition queries to an oracle are usually modelled 
so that it results in a state including the query and the answer. To end up in a similar state using quantum communication with the honest party, the adversary has to create a register that is entangled with the question. While the honest party is answering, he shares entanglement with the adversary.

In this thesis, we discuss the security definitions Boneh and Zhandry proposed in~\cite{BZ13b} and we suggest an alternative definition for quantum secure encryption, which uses apart from superposition oracle access also quantum communication without entanglement. We leave open the question whether this definition is equivalent to that of Boneh and Zhandry or strictly stronger, and whether it is feasible at all (i.e.\ whether there exist an encryption scheme that is secure in this sense). 
The technical part of the thesis is focussed on Message Authentication Codes (MACs) from pseudorandom functions under superposition attacks. 
We re-prove the theorem proved by Boneh and Zhandry~\cite{BZ13a} stating that a uniformly random function, and therefore also a pseudorandom function, can serve as a MAC which is secure, even if the adversary has superposition oracle access to this function. Formulated as a game it says that, except with negligible probability, it is impossible for any adversary to output $q+1$ input-output pairs of the oracle function after making $q$ superposition queries. 
Our approach to prove the theorem is based on the 
\emph{quantum polynomial method}~\cite{Bea97}. We follow the outline of the proof of an earlier result from quantum computing and generalize this result.

Chapter \ref{notation} gives a list of some basic notation and definitions. 
Chapter \ref{crypto} is a general introduction to cryptography and Chapter \ref{quantum} introduces some basic concepts of quantum computing and some more specific constructions that we will use in Chapter \ref{main}. 
In Chapter \ref{qsecurity} we discuss how to model quantum attacks and Chapter \ref{main} consists of the new proof of Boneh and Zhandry's theorem and some applications of it. Chapter \ref{conclusion} concludes and points out some interesting questions for future research.

\newpage
\chapter{Preliminaries}
\label{notation}
In this chapter we list some mathematical notation we use, and we define some basic concepts from probability theory and computational complexity.
We use the following notations for $n\in \mathbb{N}$:
\begin{itemize}
\item The set of all bit strings of length $n$ is denoted by $\{0,1\}^n$. \\
$\{0,1\}^*$ denotes the bit strings of any finite length.
\item $[n]$ is the set of natural numbers from 0 to $n-1$. $[n]^*$ denotes the numbers from 1 to $n-1$ (or $[n]\backslash \{0\}$).
\item $Y^X$ is the set of functions $f: X \rightarrow  Y$.
\item
The \emph{bitwise XOR} is denoted as $\oplus$ and is the operation that adds two bits modulo 2. 
Bitwise means that the operation is applied on bits in the same position. For example: $01 \oplus 11 = 10$.
\end{itemize}
\begin{definition}
A \emph{weight assignment on a set S} is a function $D: S \rightarrow \mathbb{R}$ that assigns a value to each element in the set. We can assume that the values always sum up to one (if not, we scale the values).
A \emph{probability distribution} is a weight assignment of which all values are  
non-negative and thus can be seen as probabilities. Note that in some cases we implicitly assume that all outcomes that are assigned with probability zero are excluded from the set.

When a distribution is sampled, that is: an element of the set is randomly chosen according to the distribution, we write $r \leftarrow S$. If the distribution is not explicitly mentioned, the uniform distribution is assumed. In this thesis we only consider distributions over finite sets.\\
\end{definition}

\begin{definition}
A function $\varepsilon : \mathbb{N} \rightarrow \mathbb{R}_{>0}$ is \emph{negligible} 
if $\epsilon$ decreases faster than any inverse polynomial:
$$\forall c>0 \  \exists N_c \in \mathbb{N} \text{ such that } \forall n> N_c \text{ it holds that } \varepsilon(n) < \frac{1}{ n^c}$$
\end{definition}
Equivalent to Turing Machines, every computation can be modelled by a (uniform family of) Boolean circuits of certain basic Boolean operations called \emph{gates}. 
\begin{definition} 
For every $n\in \mathbb{N}$, an \emph{$n$-input, single-output Boolean circuit} is a directed graph with $n$ \emph{sources} (vertices with no incoming edges) and one \emph{sink} (vertex with no outgoing edges). All nonsource gates are labelled with a logical operation OR, AND or NOT. The size of a circuit is the number of vertices in it. (Arora and Barak~\cite{AB09})
\end{definition}
If an algorithm has access to some oracle, then the circuit contains query gates in addition. This means that querying the oracle costs one time step.

Sets of natural numbers and formal decision problems 
can be compared by their complexity: how `difficult' it is for algorithms to compute respectively solve them. 
To which complexity class(es) a problem belongs, depends on the running time or memory space it takes an optimal algorithm to solve the problem and whether it uses randomness. 
The two best known complexity classes are the following:
\begin{itemize}
\item $\mathsf{P}$ is the class of problems that can be solved in polynomial time or equivalently with a polynomially-sized circuit. 
\item $\mathsf{NP}$ is the class of problems for which a \emph{witness} that shows the positive answer, can be verified in polynomial time. For example, the problem whether a model exists satisfying a formula $\phi$ is in $\mathsf{NP}$ because we can verify in polynomial time whether a given model satisfies $\phi$.
\end{itemize}
It has not been proved, but it is generally believed and an underlying assumption for a lot of theorems  (especially in cryptography) that  $\mathsf{P} \neq \mathsf{NP}$.

\newpage
\newcommand{\Gen}{\mathsf{Gen}}
\newcommand{\Enc}{\mathsf{Enc}}
\newcommand{\Dec}{\mathsf{Dec}}
\newcommand{\Mac}{\mathsf{Mac}}
\newcommand{\Vfy}{\mathsf{Vrfy}}
\newcommand{\Vrfy}{\mathsf{Vrfy}}
\newcommand{\negl}{\mathtt{negl}}
\newcommand{\her}{her}

\chapter{Cryptography}
\label{crypto}
This chapter provides the background in cryptography that is needed to understand Chapters \ref{qsecurity} and \ref{main}. The first section introduces the research field. In Section \ref{math} we formally define  encryption schemes and message authentication codes and the traditional security notions for these two concepts. Section \ref{pseudo} is about the more abstract notion of \emph{pseudorandomness}, that has a lot of applications in different cryptographic constructions.

\section{Introduction}
Cryptography literally means `secret writing' and has already been used in wars and diplomatic affairs from the time of Julius Caesar. 
There are two main goals in cryptography: 
\begin{itemize}
\item secrecy: an adversary who intercepts an encrypted message should not be able to gain any information about the content of the message from it.
\item authentication: it should be certain for the receiver by whom the message is sent and that the received message is exactly the one that is sent.
\end{itemize}
Assuming that the parties (often called Alice and Bob) communicate by sending data over a public channel, the only way to achieve the first goal perfectly, is  to send an encrypted version of the message (called \emph{ciphertext}) that is statistically independent from the original message, but which can be decrypted using a \emph{secret key}. This key is arranged in advance between the sender and the receiver. Statistically independent means that for all $m$ and $c$ the probability that the original message is $m$ is equal to the probability that the message is $m$, \textit{given} the cipher text is $c$.

This definition is called information-theoretic security and is only possible if the secret key is of the same size as the message and used only once, as was shown by Shannon~\cite{Sha49}. 
For example, if the length of the key is smaller than the length of the messages, it is possible to carry out a brute-force attack. The decryption algorithm is run repeatedly on the ciphertext with each possible key. Because there are less possible keys than possible messages, the outcomes of this experiment are not distributed over all messages and yield some information about the original message.
An information-theoretically secure scheme can be used in practice, but it can be difficult to secretly arrange a key each time, so it is not useful for all purposes.
In practice however, an adversary cannot always utilize the information hidden in an `insecure' ciphertext because it may take too much time to run the needed computations. 
 
Cryptographers make some assumptions about the adversaries who potentially intercept communication. By Kerckhoffs' principle, the adversary is assumed to know the fact that the sender wants to send the receiver some text and also what kind of encryption they use. 
The key is the only thing that is not known by the adversary.
As soon as a key is used multiple times, we can model the power of the adversary in several different ways. In any case an eavesdropper may use all previous ciphertexts in its computation.   Historical examples show that we may want to assume that the adversary can trick the sender into encrypting a message (or part of a message) of her choice. We call this chosen-plaintext attacks (CPA). We can go even further and assume that the adversary can trick the honest parties in encrypting \textit{and} decrypting texts of her choice: chosen-ciphertext attacks (CCA). Thus, what we call a (computationally-) secure encryption scheme  depends on the kind of attack we want to be secure against, but also on the computational power of the adversary in our model and the time we want to keep our message secret. For some applications it is not a problem if the adversary knows the message after ten years.

To be more general and be able to compare different schemes, it is convenient to use an \emph{asymptotic} approach as in complexity theory. The running time of an algorithm is treated as function of the  input length and we are interested in the asymptotic behaviour of this function, in particular whether it grows exponentially. All computations in the model should run in polynomial time. Each computation gets as additional input a string of ones whose length is called the security parameter $n$. Often it is the same as the length of the input, but it allows algorithms that not necessarily have any input at all to run for some reasonable time. The running time of the functions that are part of the scheme is then bounded by some fixed polynomial in this security parameter. We only consider adversaries that run in polynomial time in their input, which is the ciphertext, as adversaries winning the game in super-polynomial time are not considered a threat.

CPA or CCA secure encryption schemes must be randomized (apart from the randomness used for generating the key). If a message would every time be encrypted as the same ciphertext, it is easy to detect when a message is sent twice, which gives a lot of information in some cases. A random string, denoted  by $r$, can be used to choose one of a set of encryptions of the message.

There are two different types of encryption: \emph{private-key} (or symmetric) encryption and \emph{public-key} (or asymmetric) encryption. Private-key encryption is quite intuitive: two parties that want to communicate share a secret key in advance, for example when they meet physically. Public-key encryption works as follows: the receiver publicly  announces an encryption key which works in combination with a secret decryption key (known only to the receiver). This is a very useful and non-trivial concept, but not the main subject of this thesis.
\\

\noindent
The second goal of cryptography is \emph{authentication}. We want to detect 
when a message we receive (from a trusted person) is tampered with or sent by an adversary. Roughly speaking,  Message Authentication Codes (MAC) are the digital equivalence of handwritten signatures. However, the word \emph{digital signature} is reserved for the public-key version, because that concept is even more comparable to handwritten signatures which can be verified by everyone.\\

\section{Mathematical Security}
\label{math}
Where historical ciphers were just developed in a smart way which was hopefully not thought of by the adversary, modern cryptography deals with commonly known schemes where only the key is secret according to Kerckhoffs' principle. A scheme is called secure if we can prove mathematically that whatever strategy an adversary has, she cannot find out anything about the message using the information she has access to.\\

\subsection{Encryption}
Before we can define formally the security of encryption schemes, we need to define the scheme itself.
\begin{definition}
\label{encscheme}
A \emph{private-key encryption scheme} is a tuple $(\Gen, \Enc, \Dec)$ where:
\begin{itemize}
\item $\Gen$ is the key-generation algorithm having as input the security parameter $1^n$ (the bit string consisting of $n$ consecutive ones) on which it outputs a key $k$. Often this is just a sample from the uniform distribution on the set of all strings of a certain size.
\item $\Enc$ is the encryption function which, on input $k$ and a message $m \in \{0,1\}^*$, outputs a ciphertext $c$.
\item $\Dec$ is the decryption algorithm which takes as input $k$ and $c$ and outputs $m'$. 
\end{itemize}
The encryption scheme is \emph{correct} if
$\Dec$ is deterministic and $\Dec(\Enc(m))=m$ for all messages $m$.

\end{definition}
We define \emph{information-theoretic} (or \emph{perfect}) security to illustrate what relaxations are made in the definition of computational security, compared to the semantic definition of security we want to achieve.
We start with the most intuitive definition saying that an adversary can not learn anything from the ciphertext that was not already known.
\begin{definition}
An encryption scheme $(\Gen, \Enc, \Dec)$ is \emph{perfectly secret} if for every probability distribution over the message space, and every possible ciphertext we have:
$$\forall m,c: \ \Pr[M=m|C=c]=\Pr[M=m]$$
where variable $M$  is the message sent hidden in the ciphertext variable $C$, $m$ a particular message and $c$ a particular ciphertext.
\end{definition}
An equivalent way to define secrecy is to set up a game between a \emph{challenger} playing the honest parties and the adversary, which is won by the adversary if she learns more from the ciphertext than is allowed. 
\begin{definition}
An encryption scheme $(\Gen, \Enc, \Dec)$ is perfectly secret if all adversaries win the following game with probability $\frac{1}{2}$.\\
\textit{Eavesdropping game}:
\begin{enumerate}
\item The adversary $\mathcal{A}$ outputs a pair of messages $m_0,m_1 \in \{0,1\}^*$.
\item A random key is generated by running $\Gen$ and a random bit $b$ is chosen by the challenger. $c=
\Enc_k(m_b)$ is sent to $\mathcal{A}$.
\item $\mathcal{A} $ outputs a bit $b'$ and wins if $b=b'$.
\end{enumerate}
\end{definition}
It is easy to see that this notion is equivalent to the first definition. We will not work out the proof here and refer to \cite{KL07}. 
This game is referred to as $PrivK^{eav}_{\mathcal{A},\Pi}$ and is defined to be 1 if $\mathcal{A}$ wins and 0 otherwise.
We can alternatively write the adversary's final output in this game as output of the algorithm $\mathcal{A}$ on input $\Enc_k(m_b)$. The encryption scheme is perfectly secret if:
$$\Pr[\mathcal{A}(\Enc_k(m_b))=b]=\frac{1}{2}$$ $$ \Leftrightarrow \Pr[\mathcal{A}(\Enc_k(m_b))=0]- \Pr[\mathcal{A}(\Enc_k(m_b))=1] = 0$$
The probability is taken over all randomness used by the challenger, in the encryption scheme or by the adversary.
Because perfect secrecy may be too strong to be manageable and not needed for most applications,
we define several notions of a second, more relaxed type of security. We make the reasonable assumption that all agents are computationally bounded and we do not require security to last infinitely long. Hereby we arrive in the field of computational-complexity theory.
The semantic security definition with intuitive meaning in the real world, is followed by an equivalent definition in terms of indistinguishability. 

The semantic definition of computational security we state is a slightly simplified version of the more general definition in the book by Katz and Lindell \cite{KL07}. Following these ideas, stronger definitions can be made to cover adversaries that are more powerful than just eavesdroppers.
\begin{definition}
A private-key encryption scheme $\Pi = (\Gen, \Enc, \Dec)$ is \emph{semantically secure in the presence of an eavesdropper} if for any 
randomized polynomial-time adversary $\mathcal{A}$ there exists a 
randomized polynomial-time algorithm $\mathcal{A'}$ such that for all efficiently-sampleable distributions $D$ and all polynomial-time computable functions $f$, there exists a negligible function $\negl$ such that
$$\left| \Pr [\mathcal{A}(1^n, \Enc(m))=f(m)] - \Pr [\mathcal{A'}(1^n) =f(m)] \right| \leq \negl(n)$$
where $m$ is chosen according to $D$ and the probabilities are taken over the  randomness used in the encryption scheme, by the challenger and by the adversary.
\end{definition}
The intuition behind this definition is that an encryption is secure if no adversary can in practice (i.e.\ with polynomially bounded computation power) compute any partial information about the message $(f(m))$ using the eavesdropped encryption of the message, except with negligible probability.  

Because this definition is not very practical to work with, we use the following equivalent definition from which we define stronger security notions as well. For the proof of the stated equivalence we refer to the textbook by Oded Goldreich \cite{Gol04}.

\begin{definition}
A private-key encryption scheme $\Pi = (\Gen, \Enc, \Dec)$ is \emph{secure in the presence of an eavesdropper} if for any 
randomized polynomial-time adversary $\mathcal{A}$ there exists a negligible function $\negl$ such that
$$\Pr[PrivK^{eav}_{\mathcal{A},\Pi} (n) =1] \leq \frac{1}{2} + \negl(n)$$
where the random variable $PrivK^{eav}_{\mathcal{A},\Pi} (n)$ is defined as the following game:
\begin{enumerate}
\item $\mathcal{A}$ gets input $1^n$ and outputs a pair of messages $(m_0,m_1)$ of the same length.
\item $\Gen(1^n)$ generates a key $k$ and a random bit $b$ is chosen by the challenger. $\Enc_k(m_b) = c$ is given to $\mathcal{A}$.
\item If $\mathcal{A}$ outputs the bit $b$ she succeeds and the output of the game is 1, the output is 0 otherwise.
\end{enumerate}
Again, the probability is over all randomness.
\end{definition}
Real-life \emph{breaks} (i.e.\ successful attacks)  and new insights have led to refinements of the assumptions regarding the adversary. It is easy to adapt the game in the security definition in terms of indistinguishability by changing the abilities of the adversary and the information sources she has access to. We keep the semantic security definition in mind.

\begin{definition}
\label{cpa}
A private key encryption scheme $\Pi = (\Gen, \Enc, \Dec)$ has \emph{indistinguishable encryptions under a chosen-plaintext attack (or is CPA-secure)} if for any randomized polynomial-time adversary $\mathcal{A}$ there exists a negligible function $\negl$ such that
$$\Pr[PrivK^{cpa}_{\mathcal{A},\Pi} (n) =1] \leq \frac{1}{2} + \negl(n)$$
where the random variable $PrivK^{cpa}_{\mathcal{A},\Pi} (n)$ is defined as the following:
\begin{enumerate}
\item $\Gen(1^n)$ generates a key $k$
\item $\mathcal{A}$ gets input $1^n$, has oracle access to $\Enc_k$ and outputs a pair of messages $(m_0,m_1)$ of the same length.
\item A random bit $b$ is chosen. $\Enc_k(m_b) = c$ is given to $\mathcal{A}$.
\item If $\mathcal{A}$, which still has oracle access to $\Enc_k$, outputs the bit $b$ then it succeeds and the output of the game is 1, otherwise the output is 0.
\end{enumerate}
The output of $\mathcal{A}$ in this game can also be written as $\mathcal{A}^{\Enc_k}(\Enc_k(m_b))$ and the success probability is $\Pr[\mathcal{A}^{\Enc_k}(\Enc_k(m_b))=b]$ where the probability is taken over the possible keys and the randomness of $\Enc$, $b$ and $\mathcal{A}$. The superscript $\Enc_k$ means that the algorithm $\mathcal{A}$ has oracle access to $\Enc_k$.
\end{definition}
For chosen-ciphertext security we define the same game with the only addition that the adversary has oracle access to the decryption function. The adversary is not allowed to query the decryption oracle on the challenge ciphertext $c$ since this would trivially result in a break.
\begin{definition}CCA indistinguishability game $PrivK^{cca}_{\mathcal{A},\Pi} (n)$:
\begin{enumerate}
\item $\Gen(1^n)$ generates a key $k$.
\item $\mathcal{A}$ gets input $1^n$, has oracle access to $\Enc_k$ and $\Dec_k$ and outputs a pair of messages $(m_0,m_1)$ of the same length.
\item A random bit $b$ is chosen. $\Enc_k(m_b) = c$ is given to $\mathcal{A}$.
\item $\mathcal{A}$ still has access to the oracles but is not allowed to query the decryption oracle on $c$. If $\mathcal{A}$ outputs the bit $b$ then she succeeds and the output of the game is 1, otherwise the output is 0.
\end{enumerate}

\end{definition}

\subsection{Authentication}
We want to have formal security definitions for authentication as well. We start with the definition of a MAC. 
\begin{definition}
A \emph{Messsage Authentication Code (MAC)} is a triple of randomized polynomial-time algorithms $(\Gen, \Mac, \Vfy)$:
\begin{itemize}
\item $\Gen$ outputs a key $k$ on input $1^n$, $|k|\geq n$ (the security parameter).
\item $\Mac$ takes input $m \in \{0,1\}^*$ and $k$ and outputs a tag $t$.
\item $\Vfy$ checks validity of the tag: on input $m$, $k$ and $t$ it outputs 0 or 1. 
\end{itemize}
The authentication scheme is \emph{correct} 
if indeed $\Mac_k(m)=t$ (for some randomness $r$) when $\Vfy_k(m,t) =1$ and $\Mac_k(m)\neq t$ (for all $r$) when $\Vfy_k(m,t) =0$.

\end{definition}
The definition of secure MACs needs to formalize the requirement that a message from a trusted sender cannot be changed by an adversary without being noticed by the receiver. This requirement immediately covers infeasibility   of sending a message as if it is sent by this particular trusted party. For MACs there is only one definition used in general (and variations of it), which is security against the adaptive chosen-message attack. 

\begin{definition}
A MAC is \emph{existentially unforgeable under adaptive chosen-message attacks} if for any adversary the probability that the adversary wins the following game is negligible.\\
\emph{MAC-forge game}:
\begin{enumerate}
\item  $\Gen(1^n)$ generates a random key.
\item $\mathcal{A}$ is given input $1^n$ and oracle access to $\Mac_k$. $\mathcal{A}$ wins if she outputs a pair $(m,t)$ for which it holds that $m$ was not queried and $\Vrfy_k(m,t)=1$.
\end{enumerate}
\end{definition}
We can equivalently say that the adversary wins if she outputs $k+1$ distinct valid message-tag pairs where $k$ is the number of queries made. We will use this formulation to define quantum security in Chapter \ref{qsecurity}.

Note that it is possible for an adversary to send a valid message-tag pair of a MAC that is secure, according to this definition, 
 by resending a pair that was sent by a trusted sender. A MAC can be protected against these \emph{replay attacks} by requiring a sequence number or a time-stamp as part of  each message, but we will not discuss this further.
Roughly the same security definition works for public-key authentication: signature schemes. The difference in the game is that everyone, including the adversary, can check the validity of a signature using the public key.

\section{Pseudorandomness}
\label{pseudo}
A cryptographic scheme is secure against a certain class of adversaries if ciphertexts  `look random' to all adversaries. This is exactly what pseudorandomness achieves: making much pseudorandom output from little random input.   Apart from that, pseudorandomness can be used as alternative for real randomness. 
One of the applications, and one we study in this thesis, is the almost trivial construction of MACs  from \emph{pseudorandom functions} (PRF).  PRF's can be made from \emph{pseudorandom generators}, which in their turn can be build using \emph{one-way functions}.

\begin{definition} 
A \emph{pseudorandom function} is an efficient length-preserving keyed function $f \colon \{0,1\}^* \times \{0,1\}^* \rightarrow \{0,1\}^*$ such that for all 
randomized polynomial-time algorithms $\mathcal{D}$, there exists a negligible function $\negl$ such that:
$$\left| \Pr[\mathcal{D}^{f_k}(1^n)=1]- \Pr[\mathcal{D}^{g}(1^n)=1] \right| \leq \negl(n)$$
where $k$ and the function $g$ are chosen uniformly at random. 
\end{definition}

\begin{definition}
A \emph{pseudorandom generator} is an efficient deterministic algorithm $G$ which, on input $s\in \{0,1\}^n$, outputs a string of length $l(n)$ such that:
\begin{itemize}
\item The \emph{expansion factor} $l(n)$ is a polynomial in $n$ with $\forall n \ l(n)>n$.
\item No randomized polynomial-time algorithm $\mathcal{D}$ can distinguish $G(s)$ from a uniformly random string of the same length:
$$|\Pr[\mathcal{D}(r)=1]-\Pr[\mathcal{D}(G(s))=1]|\leq \negl(n)$$  
where $s$ and $r$ are uniformly random strings of respective sizes $n$ and $l(n)$, and $\negl$ is some negligible function (that may depend on $\mathcal{D}$).
\end{itemize}
\end{definition}

\begin{definition}
A \emph{one-way function} is a function $f\colon \{0,1\}^* \rightarrow \{0,1\}^*$ which is easy to compute and hard to invert:
\begin{itemize}
\item There exists a polynomial-time algorithm $M_f$ such that $ \forall x \ M_f(x)=f(x) $.
\item Any randomized polynomial-time $\mathcal{A}$ wins the following \emph{inverting (or collision finding) game} only with negligible probability:\\
	$\mathcal{A}$ is given  $1^n$ and $f(x)$ of a random $x \leftarrow \{0,1\}^n$ and outputs $x'$. $\mathcal{A}$ wins if $f(x')=f(x)$
\end{itemize}
\end{definition}

\newpage
\newcommand{\xv}{\mathbf{x}}
\newcommand{\M}{m}
\newcommand{\bfn}{\mathbf{b}}

\chapter{Quantum Computing}
\label{quantum}
This chapter provides some background on quantum computing. Section \ref{qcintro} is an introduction to  quantum computers. Section \ref{tech} is about the formal definition and the commonly used notation to write down computations concisely. Section \ref{secepr} describes the famous EPR-pair \cite{EPR} illustrating the phenomenon \emph{entanglement}. In Section \ref{querymodel} we introduce the \emph{query model} in which a function is modelled as black box or \emph{oracle} instead of a subroutine of an algorithm. The algorithm may query a \emph{quantum oracle} on any superposition of inputs. We prove the equivalence of two different implementations of oracle queries to any (non-boolean) function. This shows that the implementation we work with in Chapter \ref{main} is equivalent to the implementation Boneh and Zhandry use~\cite{BZ13a}.

\section{Introduction}
\label{qcintro}
To describe small particles one should use the theory of quantum mechanics. Although modern elements of computer chips are too small to behave classically, the unwanted quantum mechanical behaviour is `corrected' because we want computers to work like we expect and according to which classical software is designed. Theoretically, it is also possible to exploit the quantum effects and build a computer based on these phenomena instead of classical logic. In practice it is shown to be possible to build a \textit{quantum computer} consisting of very few \textit{quantum bits} (the analogue of a classical bit) and researchers are working on other methods that are hopefully more scalable. It is not clear whether it is possible at all or when a quantum computer will be built with computing power comparable to state-of-art classical computers, but when the time comes, we can no longer expect our adversaries to use classical computers only. 

The theoretical aspects of quantum computation have been studied since the 1980s, and discoveries have been made that are of great importance for cryptography. The bad news is that,  using Shor's surprisingly efficient quantum algorithm for factorization, it is easy to break the widely used RSA public-key encryption schemes. The good news is there exist information-theoretically secure encryptions schemes that make use of quantum computation and communication to distribute keys. However,  when quantum computers first become usable, they will be scarce and expensive and honest parties who want to communicate securely will in general not have access to a quantum computer, while criminal organizations may have. In this situation it is important to use \emph{classical} cryptographic schemes proven to be secure against \emph{quantum} attacks. The research field that investigates the vulnerability of existing or new \emph{classical} schemes to quantum power, is called Post-Quantum Cryptography. A lot of (symmetric) schemes that are proven to be classically secure are expected to be \emph{quantum secure} as well, maybe under some additional conditions, but new proofs and/or proof techniques are required.\\

\noindent
Three important phenomena in quantum mechanics are \emph{superposition}, \emph{interference} and the \emph{collapse after measurements}. Superposition  means a system can be in more than one state at a time, each with some \emph{amplitude} (a complex coefficient). Due to interference, positive and negative amplitudes can cancel each other out. In quantum computing a qubit can be both 0 and 1 at the same time, but after a measurement the qubit collapses into a classical \textit{basis} state. The probability that the qubit is measured as 0 is the squared modulus of the amplitude of this state. The squared moduli of all amplitudes must sum up to one. A system or \textit{register} of $n$ qubits can be in $2^n$ states at the same time and algorithms on this register may compute a function on $2^n$ different inputs simultaneously. 
However, it is not possible to see the outcomes of all these parallel computations because the register will collapse to a single basis state after measurement. Only if interference can be used in a smart way, can the measurement 
tell us something about multiple inputs. 

There are only a few specific problems for which there exists a quantum algorithm solving it, which is much (exponentially) faster than any known classical algorithm that solves the problem.

In the field of quantum computing, people often make use of the circuit model of algorithms, which is, like the classical model, equivalent to the model of (quantum) Turing machines.

\section{Technical Framework}
\label{tech}
The state of a qubit can be written as two-dimensional vector 
$\left( \begin{array}{c}
\alpha\\
\beta\\
\end{array}
\right)$
in a Hilbert space, where $\alpha$ is the amplitude of basis vector 
$\left( \begin{array}{c}
1\\
0\\
\end{array}
\right)$
and $\beta$ the amplitude of 
$\left( \begin{array}{c}
0\\
1\\
\end{array}
\right)$.
$\alpha$ and $\beta$ are complex numbers and their squared absolute values sum up to 1. 
We use the more concise Dirac-notation, which is conventional in both quantum mechanics and quantum computing. In this notation, a column vector $\phi =  
\left( \begin{array}{c}
\alpha\\
\beta\\
\end{array}
\right)$
is written as \emph{ket} $|\phi\rangle = \alpha |0\rangle + \beta |1\rangle$ and a row vector $\psi$ as \emph{bra} $\langle \psi |$. The product \emph{bra $\cdot$ ket} $\langle \psi | \phi \rangle$ is the inner product of the two vectors, which corresponds with the widely used notation for inner product.

Every operation that can be applied on a quantum state preserves the following property of a quantum state: 
the sum of the squared amplitudes is 1. In terms of linear algebra, this means that operations or circuit gates on a quantum state can be written as unitary matrices left-multiplied with the state.

The tensor product $\otimes$ is used to combine multiple qubits or registers of qubits in one big quantum system. The set of possible classical states of a combined system is the Cartesian product of the state sets of each part.  

The combined state of two qubits $|\phi\rangle$ and $|\psi\rangle$ is written as $|\phi\rangle \otimes |\psi\rangle$ and sometimes abbreviated to $|\phi\rangle |\psi\rangle$, $|\phi , \psi\rangle$ or $|\phi \psi\rangle$. The bit string that arises in this way, in case of basis states, can be written as number between 0 and $2^n-1$ where $n$ is the length of the bit string.
Two parallel unitary transformations on two registers is equal to the tensor product of the operations applied on the tensor product of the registers.\\

\noindent
When we measure a register of $n$ qubits in the state $|\phi\rangle = \sum_{j=0}^{2^n-1} \alpha_j |j\rangle$, the probability of seeing the classical state $|j\rangle$ is $|\alpha_j|^2$. We say that the quantum state has \emph{collapsed} to the classical state $|j\rangle$ after the measurement. All other information the quantum state held is lost. This standard implementation is called \emph{measurement in the computational basis} but there are more possibilities. In general, measurements in any orthonormal basis can be described as a \emph{projective measurement}. For more information about this see for example the first chapter of Ronald de Wolf's PhD thesis \cite{Wol01} or the textbook by Nielsen and Chuang \cite{NC00}.

The circuit model for classical computation uses logical gates as NOT, AND, OR, XOR. There are different minimal sets of gates which generate all other gates. The same is true for quantum computing if we allow a negligible error probability.

The following quantum gates are commonly used:\\
 $$X=
\left( \begin{array}{cc}
0&1\\
1&0\\
\end{array}
\right) \ \ 
Z=
\left( \begin{array}{cr}
1&0\\
0&-1\\
\end{array}
\right) 
$$ 
Bitflip gate X is the quantum equivalent of NOT and Z is a phaseflip gate.\\
The Hadamard gate changes a classical state in a `very' quantum state: it sends both $|0\rangle$ and $|1\rangle$ to a uniform superposition over the two states. The phase flip on the amplitude of $|1\rangle$ makes the difference.\\
$$H= \frac{1}{\sqrt{2}}
\left( \begin{array}{cr}
1&1\\
1&-1\\
\end{array}
\right)
$$\\
  $$CNOT =
\left( \begin{array}{cccc}
1&0&0&0\\
0&1&0&0\\
0&0&0&1\\
0&0&1&0\\
\end{array}
\right)
$$ The controlled-not gate negates the second bit if the first bit is 1 and does nothing otherwise.\\

\section{Entanglement}
\label{secepr}
An important phenomenon that illustrates the counter-intuitive character of quantum mechanics is entanglement. Possibly separated particles can be related such that they seem to `communicate instantly' over any distance.  However, this is impossible because no matter or information can travel faster than light. Something is going on that looks like communication but it is something else. 

The following circuit creates two qubits being `fully entangled': an EPR-pair (named after a famous paper by Einstein, Podolsky and Rosen~\cite{EPR}).\\
\begin{figure}[h!]
\label{epr}
\begin{center}
\includegraphics[scale=1.5]{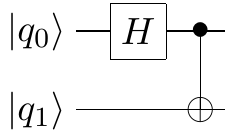}
\end{center}
\end{figure}\\
We calculate the state after this circuit, applied on $|00\rangle$. The Hadamard gate on the first bit transforms it to $\frac{1}{\sqrt{2}}(|0\rangle +|1\rangle)$. Now the CNOT gate flips the second bit to 1 if the first bit is 1. This means that either both bits are 1 or both bits are 0. The state we have is $\frac{1}{\sqrt{2}} ( |00\rangle + |11\rangle)$.
Suppose that Alice holds the first qubit and the second qubit is given to Bob who takes it to a place far away from Alice. If Alice now measures her qubit and sees (say) 1, then Bob's qubit collapses at the same moment: if Bob measures his qubit, he will see 1 as well.  This could not have happened as a causal relation, since it happened faster than anything sent by Alice could reach Bob.  This  contradicts to relativity theory unless (at least) one of the following two assumptions (which we tend to make intuitively) is dropped. 
\begin{itemize}
\item \emph{realism}: the physical properties that objects have are independent of observation.
\item \emph{locality}: measurements can not influence the outcome of other measurements from a distance.
\end{itemize}
Although Alice and Bob can share entanglement, they cannot send each other information faster than light. When Alice measures her qubit, she cannot choose to which state her or Bob's qubit will collapse. Neither can Bob.
\section{Query Model}
\label{querymodel}
Instead of time complexity, quantum algorithms are often analysed in terms of query complexity. If we know how many queries an algorithm needs to make to some function or random-access memory, we have a lower bound on the computing time as well. Each query costs one time step. 

The query model is used in cryptography as well. In the security definitions we assumed adversaries to have oracle access to the encryption, decryption or MAC-function. A lower bound on the number of queries an adversary needs to break a scheme immediately implies a lower bound on the time the break takes, so query complexity is useful in this field too.

The difference between the two fields is that cryptographers treat oracles as functions: the query is an element of the domain of the function and the answer is an element of the range, whereas in quantum computing oracles are usually memory of which the algorithm queries a bit by its index. It is not a big difference since a bit string is just a Boolean function, and any function can be represented as string. However, we have to be careful with the details. The main theorem in Chapter \ref{main} is about a cryptographic non-Boolean oracle, which is the reason why we have to prove a standard technique used in quantum computing for this general case (Proposition \ref{poly}). 

A quantum algorithm may query the oracle on a superposition of inputs.
There are different possibilities to implement an oracle answering a \emph{superposition query}. We describe two commonly used implementations and prove their equivalence.

\begin{proposition}
\label{querytypes}
Let $f\colon X \rightarrow Y$ be the function to which an algorithm has oracle access, $|X|= n,\ |Y| = m $.
The following two transformations on a register of qubits are equivalent formalizations of a single superposition query. The register consists of three parts. The first part $x$ is used for the query, $b$ for the answer and $w$ is workspace.
\begin{itemize}
\item addition query $O_f: |x,b,w \rangle \mapsto |x, b + f(x),w \rangle $ where $+$ is the operation of the group in which $b$ and $f(x)$ live, in this case addition modulo ${\M}$.
\item phase query $O_f^{\circ}: |x,b,w \rangle \mapsto \omega_{\M}^{b \cdot f(x)} |x,b,w \rangle$ and  $\omega_{\M}$ is the ${\M}$th complex root of unity $e^{2\pi \frac{i}{{\M}}}$.
\end{itemize}
\end{proposition}
\begin{proof}
We prove equivalence by giving two circuits that implement one formalization using the other.  These circuits show  that we get one type of queries by viewing the operation of the other type in a different basis, since the only thing we do is change to the `Fourier basis' to make the query and then change back to the computational basis.\\
We start with the implementation of a phase query using an addition query.\\

\noindent
\textbf{Circuit:}
\begin{itemize}
\item The inverse Fourier transform is applied on first register $|b\rangle$.
\item The transformation of the \emph{addition query} is applied on the total state.
\item The Fourier transform is applied on the first register.\\
\end{itemize}
By $F_{\M}$ we denote the ${\M}^{th}$ Fourier transform which is an $m$-by-$m$ matrix with entries $\frac{1}{\sqrt{{\M}}} \omega_{\M}^{jk}$ where $j$ is the row and $k$ the column of the entry. The inverse of the Fourier  transform $F_m^{-1}$ is defined similarly: each entry is $\frac{1}{\sqrt{{\M}}} \omega_{\M}^{-jk}$. From here on, we often omit the ${\M}$ if it is clear which root is meant.\\

\noindent
At the start of the circuit, the registers of the algorithm are in the following state:
\begin{align*}
&\sum_{x,b,w} \alpha_{xbw} |x,b,w\rangle
\\
\intertext{After the application of the inverse Fourier transform on $|b\rangle$ each basis state $|x,b,w \rangle$ becomes $\frac{1}{\sqrt{{\M}}} \sum\limits_{j=0}^{{\M}-1} \omega_{\M}^{-bj} |x,j,w\rangle$, so
the total state is:}
&\sum_{x,b,w} \alpha_{xbw} \frac{1}{\sqrt{{\M}}} \sum_{j=0}^{{\M}-1} \omega_{\M}^{-bj} |x,j,w\rangle
\\
\intertext{On this state the transformation of the \emph{addition query} is applied which results in the following state:}
&\sum_{x,b,w} \alpha_{xbw} \frac{1}{\sqrt{{\M}}} \sum_{j=0}^{{\M}-1} \omega_{\M}^{-bj} \ |x,j+f(x),w\rangle \hspace{2cm}
\end{align*}
Then after the Fourier transform we get:
\begin{align*}
 \sum_{x,b,w} \alpha_{xbw} \frac{1}{\sqrt{{\M}}} \sum_{j=0}^{{\M}-1} \omega_{\M}^{-bj} \ 
\frac{1}{\sqrt{{\M}}} \sum\limits_{k=0}^{{\M}-1} \omega_{\M}^{k(j+f(x))} |x,k,w\rangle
\\
= \frac{1}{{{\M}}} \sum_{x,b,w} \alpha_{xbw}  \sum_{j=0}^{{\M}-1}  
 \sum\limits_{k=0}^{{\M}-1} \omega_{\M}^{-bj} \omega_{\M}^{k(j+f(x))} |x,k,w\rangle
\\
= \frac{1}{{{\M}}} \sum_{x,b,w} \alpha_{xbw}  \sum_{j=0}^{{\M}-1}  
 \sum\limits_{k=0}^{{\M}-1} \omega_{\M}^{j(k-b)+k\cdot f(x)}
  |x,k,w\rangle
\\
= \frac{1}{{{\M}}} \sum_{x,b,w} \alpha_{xbw}  \sum_{k=0}^{{\M}-1}  
 \sum\limits_{j=0}^{{\M}-1} \omega_{\M}^{j(k-b)+k\cdot f(x)}
  |x,k,w\rangle 
\end{align*}
$$
= \frac{1}{{{\M}}} \sum_{x,b,w} \alpha_{xbw} \left( \sum_{\substack{k=0\\ k\neq b}}^{{\M}-1}  
 \sum\limits_{j=0}^{{\M}-1} \omega_{\M}^{j(k-b)+k\cdot f(x)}
  |x,k,w\rangle
  +  \sum\limits_{j=0}^{{\M}-1} \omega_{\M}^{j(b-b)+b\cdot f(x)}
  |x,b,w\rangle \right)
$$
We split the sum over $k$ into a single term $k=b$ and the sum over all other terms. The reason for this is that all terms $k\neq b$ have the form
$$\sum\limits_{j=0}^{{\M}-1} \omega_{\M}^{j(k-b)+k\cdot f(x)}
  |x,k,w\rangle  = 
  \sum\limits_{j=0}^{{\M}-1} \omega_{\M}^{jc+k\cdot f(x)}
  |x,k,w\rangle $$
  for some non-negative constant c.
By Proposition \ref{zerosum} (below), each of these sums over $j$ equals 0. 
So we are only left with the term $k=b$ which is exactly our definition of the \emph{phase query}:
$$   \sum\limits_{j=0}^{{\M}-1} \omega_{\M}^{j(b-b)+b\cdot f(x)}
  |x,b,w\rangle 
=  \sum\limits_{j=0}^{{\M}-1} \omega_{\M}^{b\cdot f(x)}
  |x,b,w\rangle 
$$
The second part of the proof is to show that we can implement the addition query using a phase query. The circuit is very similar to the previous one, namely the Fourier transform and its inverse have swapped places: the Fourier transform is applied before the query and the inverse Fourier transform after the query.
We start again with some arbitrary state: 
\begin{align*}
&\sum_{x,b,w} \alpha_{xbw} |x,b,w\rangle\\
\intertext{The Fourier transform is applied to the $b$-register of this state which turns it into:}
&\sum_{x,b,w} \alpha_{xbw} \frac{1}{\sqrt{{\M}}} \sum_{j=0}^{{\M}-1} \omega_{\M}^{bj} |x,j,w\rangle\\
\intertext{Then a query is made implemented as phase change, resulting in:}
&\sum_{x,b,w} \alpha_{xbw} \frac{1}{\sqrt{{\M}}} \sum_{j=0}^{{\M}-1} \omega^{bj} \omega^{j\cdot f(x)}|x,j,w\rangle\\
\end{align*}
And after the inverse Fourier transform we have the state:
\begin{align*}
&\sum_{x,b,w} \alpha_{xbw} \frac{1}{\sqrt{{\M}}} \sum_{j=0}^{{\M}-1} \omega^{bj} \omega^{j\cdot f(x)}
\frac{1}{\sqrt{{\M}}} \sum_{k=0}^{{\M}-1} \omega^{-jk}
|x,k,w\rangle
\\
= \frac{1}{{{\M}}} &\sum_{x,b,w} \alpha_{xbw}  \sum_{k=0}^{{\M}-1} \sum_{j=0}^{{\M}-1} \omega^{bj} \omega^{j\cdot f(x)}
  \omega^{-jk}
|x,k,w\rangle
\end{align*}
$$= \frac{1}{{{\M}}} \sum_{x,b,w} \alpha_{xbw}  \left( \sum_{\substack{k\neq 0\\k=b+f(x)}}^{{\M}-1} \sum_{j=0}^{{\M}-1} \omega^{j(b+f(x)-k)} |x,k,w\rangle
+ \sum_{j=0}^{{\M}-1} \omega^{j\cdot 0} 
|x,b+f(x),w\rangle \right)$$

Using  Proposition \ref{zerosum} we get:
\begin{align*}
\frac{1}{{{\M}}}& \sum_{x,b,w} \alpha_{xbw}  \left( \sum_{\substack{k\neq 0\\k=b+f(x)}}^{{\M}-1} 0
+ \sum_{j=0}^{{\M}-1} \omega^{0} 
|x,b+f(x),w\rangle \right)
\\
 = \frac{1}{{{\M}}}& \sum_{x,b,w} \alpha_{xbw} \cdot {\M} \ 
|x,b+f(x),w\rangle 
\\
 = & \sum_{x,b,w} \alpha_{xbw} \ 
|x,b+f(x),w\rangle 
\end{align*}
\qed
\end{proof}

\begin{proposition}
\label{zerosum}
Let $\omega_m = e^{2\pi \frac{i}{m}}$ for  $m\in \mathds{N}$ and let $c$ be a non-zero integer that is not a multiple of $m$ ($m\nmid c$). 
Then
$$
\sum_{j=0}^{{\M}-1} \omega_{\M}^{jc} =0
$$
\end{proposition}
\begin{proof}
It follows directly from well known formula for the sum of a geometric series, but it gives more insight to see the proof.
For $c=1$ 
one can easily see that the expectation ranges over exactly all distinct powers of omega on the unit circle in the complex plane. 
Viewing these points, it is intuitively clear that they sum op to 0 because they are evenly distributed on the circle. 
We can show this algebraically for arbitrary $c$ by showing that the set of points and thus their sum  stays the same if we rotate everything by an angle of $\frac{2\pi c}{{\M}}$.\\
\begin{align*}
e^{\frac{2\pi c}{{\M}} i} \cdot \sum_{j=0}^{{\M}-1} e^{2\pi c \frac{i}{{\M}}\cdot j} & = \sum_{j=0}^{{\M}-1} e^{2\pi c \frac{i}{{\M}}\cdot j} \cdot e^{\frac{2\pi c}{{\M}} i}\\
& = \sum_{j=0}^{{\M}-1} e^{2\pi c \frac{i}{{\M}}\cdot ( j+1)}\\
& = \sum_{k=1}^{{\M}} e^{2\pi c \frac{i}{{\M}}\cdot k}\\
& = \sum_{k=1}^{{\M}-1} e^{2\pi c \frac{i}{{\M}}\cdot k} + e^{2\pi c  \frac{i}{{\M}}\cdot {\M}}\\
& = \sum_{k=1}^{{\M}-1} e^{2\pi c \frac{i}{{\M}}\cdot k} + 1\\
& = \sum_{k=1}^{{\M}-1} e^{2\pi c \frac{i}{{\M}}\cdot k} + e^{2\pi c  \frac{i}{{\M}} \cdot 0 }\\
& = \sum_{k=0}^{{\M}-1} e^{2\pi c \frac{i}{{\M}}\cdot k}
\end{align*}
If a multiplication (of elements in a field) does not change an object, then either the object is multiplied by (the identity) one, or the object is zero.
Because $e^{\frac{2\pi c}{{\M}} i} \neq 1$ if $c$ is not a multiple of $m$, it follows that the sum must be 0.
\qed\\
 
\end{proof}

\newpage
\chapter{Quantum Security Models}
\label{qsecurity}
In this chapter we discuss different ways to model a quantum attacker.
In any case a \emph{quantum adversary} can use a quantum computer to run quantum algorithms. We know that for some public-key schemes, the key can be learned using a quantum algorithm. 
When RSA is used for example, the adversary can compute the prime numbers (secret key) from their product (public key).
Of course quantum adversaries have access to the same sources of information or oracles as their classical counterparts. When we are reasoning about communication and oracle queries in situations in which there are classical computers as well as quantum computers, the communication can be quantum or classical. 
The choices we make for our security models depend on how adversaries in reality gain certain information. 

For example, in the random oracle model, the random oracle would in reality be replaced by some hash function which can be computed by the adversary  itself. Because we can never be sure how an adversary implements the hash function, we  assume in our model that the adversary can compute it in superposition. 
Analogously, 
the random oracle may be queried in superposition by the adversary. Boneh et al. \cite{BDFLSZ11} show that security in this \emph{quantum random-oracle model} is harder to prove for schemes in general. They even construct a scheme that is secure if the random-oracle is queried only classically, but insecure if  it is queried in superposition. 

Other oracles that occur in security definitions are the pseudorandom and uniformly-random functions in distinguishing games, encryption and decryption oracles, and signing oracles. In the line of random-oracle proofs, it is interesting to look at pseudorandom functions since they can be used to simulate a random oracle. In general, \emph{quantum secure} pseudorandom functions can be used for several things if one wants to have a conservative model with minimized limitations on the adversary's abilities. 
Zhandry \cite{Zha12b} shows that quantum-secure pseudorandom functions are needed to simulate quantum-accessible random oracles if the number of queries is not bounded in advance. 
Here \emph{quantum-accessible random oracle} means that the adversary can query the oracle on a superposition of states. This can be implemented in several ways as explained in Chapter \ref{quantum}. 
Equivalently, a PRF is quantum secure if it cannot be distinguished from a uniformly random function by an adversary making quantum queries to the oracle. 

Going further in giving adversaries superposition oracle access becomes more difficult when we consider encryption and authentication, because the security games are not trivial to translate. While having only the oracle to consider and no other input (except for the security parameter) in case of pseudorandom functions, it is getting more complex when we model encryption games in which adversaries have  communication  back-and-forth. Which inputs does the adversary get in superposition?

\section{Encryption}
Recall the security game of a chosen-plaintext attack on an encryption scheme. (Definition \ref{cpa})
\begin{enumerate}
\item $\Gen(1^n)$ generates a key $k$.
\item $\mathcal{A}$ gets input $1^n$, has oracle access to $\Enc_k$ and outputs a pair of messages $(m_0,m_1)$ of the same length.
\item A random bit $b$ is chosen. $\Enc(m_b) = c$ is given to $\mathcal{A}$.
\item If $\mathcal{A}$, which still has oracle access to $\Enc$, outputs the bit $b$ then it succeeds and the output of the game is 1, otherwise the output is 0.
\end{enumerate}
To turn this into a \emph{superposition-chosen-plaintext attack}, the first option is to make the oracle access to $\Enc_k$ quantum. That is, the adversary can query the oracle on a superposition of inputs, getting back a superposition of answers.

To see how useful this definition is, 
we look at the real-life scenario this definition aims to model. Usually, in the field of post-quantum cryptography, it is assumed that the adversary has a quantum computer and the honest parties have only a classical machine. If both the sender and the receiver had quantum computers, they would be able to use quantum cryptography. However, it is possible that Alice, having a quantum computer, chooses to use a classical scheme to communicate to several others (maybe without a quantum computer) and 
has found an implementation of the scheme for her quantum computer. If the adversary finds a way to get the final superposition of this quantum encryption algorithm just before Alice measures it to send it classically, then the scenario fits with the adversary getting superposition answers. To be conservative, we assume that the adversary can choose any superposition to query in this way (however unrealistic this may sound). 
Assuming that the adversary can obtain the final superposition of the encrypted queries, we cannot always exclude that the adversary can do the same for all messages that are sent by Alice. We want to have the choice to make this assumption or not. Allowing Alice to see all superpositions of ciphertexts means that the challenge ciphertext $c_b$ would also be in some superposition.

Boneh and Zhandry \cite{BZ13b} formalize this (as a first try) in the following way, without much discussion about why this is a reasonable option.
\begin{definition} \cite[Definition 4.1]{BZ13b}
A private-key encryption scheme $\Pi = (\Gen, \Enc, \Dec)$ is \emph{indistinguishable under a fully quantum chosen-plaintext attack (IND-fqCPA secure)} if no efficient adversary $\mathcal{A}$ can win the following game, except with probability at most $\frac{1}{2} + \negl$:
\begin{enumerate}
\item A key $k$ is generated using $\Gen$ and a random bit $b$ is chosen.
\item $\mathcal{A}$ is allowed to make chosen-message queries on superpositions of message pairs.  For each such query, the challenger chooses randomness $r$, and encrypts the appropriate message in each pair using $r$ as randomness:
$$ \sum_{m_0,m_1,c} \psi_{m_0,m_1,c} |m_0,m_1,c\rangle \rightarrow \sum_{m_0,m_1,c} \psi_{m_0,m_1,c} |m_0,m_1,c \oplus \Enc_k(m_b;r)\rangle$$ 
\item $\mathcal{A}$ produces a bit and wins the game if the bit is equal to $b$.
\end{enumerate}
\end{definition}
Because every oracle answer contains information about $b$, this adversary is very powerful. Boneh and Zhandry prove that there cannot exist a scheme that satisfies this definition, because the message query is entangled with the answer.
They try to solve this problem by changing the implementation of the queries: both messages will be encrypted, but depending on $b$ the order is flipped in the answer. Unfortunately, this definition does not solve the problem: they prove that this definition is at least as strong as the first one.
They proceed with the same game we mentioned above. In this game \ref{INDqCPA}, the challenge ciphertext is one of $\Enc(m_0)$ and $\Enc (m_1)$ where the classic pair $(m_0,m_1)$ is chosen by the adversary.\\
\begin{definition} \cite[Definition 4.5]{BZ13b}
\label{INDqCPA}
A private-key encryption scheme $\Pi = (\Gen, \Enc, \Dec)$ is \emph{indistinguishable under a quantum chosen-plaintext attack (IND-qCPA secure)} if no efficient adversary $\mathcal{A}$ can win the following game, except with probability at most $\frac{1}{2} + \negl$:
\begin{enumerate}
\item A key $k$ is generated using $\Gen$, and a random bit $b$ is chosen.
\item $\mathcal{A}$ is allowed to make:
	\begin{itemize}
	\item \emph{challenge queries}: $\mathcal{A}$ sends a pair $(m_0,m_1)$ and gets back $c^*= \Enc(k,m_b)$.
	\item \emph{encryption queries}:  
	For each such query, the challenger chooses randomness $r$, and encrypts each message using $r$ as randomness:
$$ \sum_{m,c} \psi_{m,c} |m,c\rangle \rightarrow \sum_{m,c} \psi_{m,c} |m,c \oplus \Enc_k(m;r)\rangle$$ 
	\end{itemize}
\item $\mathcal{A}$ produces a bit and wins the game if the bit is equal to $b$.
\end{enumerate}
\end{definition}
What is reasonable to allow the adversary in order to model the real world? 
To answer this question, we recall the ideas that led to the traditional definition.
The adversary has access to the encryption oracle to model the ability of tricking an honest party into sending particular messages. We want to assume, in this new model, that this tricking can be done in superposition. It can be the case that a copy of an encryption device is in hands of the adversary. This device works as a black box, so  it can be run on superpositions of  messages without revealing the key.

The message pair $(m_0,m_1)$ of the classical distinguishing game is introduced to have two situations that 
the adversary can try to distinguish. This distinguishing ability should be equivalent or close to the ability of learning some information about the  messages Alice genuinely sends to Bob. To prove that whatever Alice wants to send, the adversary learns nothing (new) from the ciphertext, we let the adversary choose the pair (classically) in the game. We have to keep in mind that in reality there is no communication between the adversary and the sender; the communication with the \emph{challenger} exists in the model to ensure that the scheme is secure for any such pair. 
To cover all 
scenarios, we do want to assume that the adversary can always get the sender's superposition just before it is to be measured and sent.  
Note that the sender loses the quantum state in this case and the sending process is aborted. 
We do not discuss the question why the sender wants to send a measured superposition of messages; we just want to be general. Even if the definition is too strong for most uses, we want to have the choice to use it instead of weaker ones. 
Above all it is interesting to compare different models.

The semantic notion of security we want to achieve is the following: no adversary can learn anything about the superposition of messages the sender started with, from the superposition of ciphertexts the sender was about to measure and send. Our goal is to state this notion as distinguishing game.

We propose to give the adversary the task to distinguish between two superpositions of ciphertext. As in the classical case,  we let the adversary choose the two superpositions because it has to be safe for any message the sender may send. The formalization we use (Definition \ref{newdef}) is not, like Boneh and Zhandry do, the same as the implementation of the queries, in which the query itself is still part of the state afterwards. In our model the adversary sends a pair of superpositions over a quantum channel and receives one of the two superpositions of ciphertexts, depending on $b$. 
We require that the pair of superpositions is not entangled. 
This may seem questionable, but since the pair does not model any real state prepared by the adversary, it is reasonable to make this requirement.
Note that a special case is the situation in which the adversary chooses two classical states. 
Because there is no entanglement between different registers of the adversary,  we are (at least at first sight) not facing the problem that occurs with the mentioned option Boneh and Zhandry proposed.
It is easy to see that this definition is stronger than (what they call) \textit{IND-qCPA security} because the last one is a special case. Whether it is feasible at all and if so, whether it is strictly stronger than \textit{IND-qCPA security} is an interesting open question.
\begin{definition}
\label{newdef}
A private-key encryption scheme $\Pi = (\Gen, \Enc, \Dec)$ is \emph{indistinguishable under a superposition chosen-plaintext attack (IND-sCPA secure)} if no efficient adversary $\mathcal{A}$ can win the following game, except with probability at most $\frac{1}{2} + \negl$:
\begin{enumerate}
\item A key $k$ is generated using $\Gen$ and a random bit $b$ is chosen.
\item $\mathcal{A}$ is allowed to make \emph{encryption queries}.  
	For each such query, the challenger chooses randomness $r$, and encrypts each message using $r$ as randomness:
$$ \sum_{m,c} \psi_{m,c} |m,c\rangle \rightarrow \sum_{m,c} \psi_{m,c} |m,c \oplus \Enc_k(m;r)\rangle$$ 
\item $\mathcal{A}$ chooses two non-entangled superpositions of messages $(|m_0\rangle, |m_1\rangle)$ and sends them over a quantum channel to the challenger. (Non-entangled means in this case that after sending the pair, the adversary does not hold qubits that are entangled with $(|m_0\rangle \text{ or } |m_1\rangle)$. Also, $(|m_0\rangle \text{ and } |m_1\rangle)$ are not entangled with each other.)
\item The challenger chooses randomness $r^*$. The messages in $|m_b\rangle$ are encrypted using $r^*$ and sent to $\mathcal{A}$:
$$ |m_b\rangle = \sum_m \psi_m |m\rangle \rightarrow \sum_m \psi_m |\Enc_k(m;r^*)\rangle$$
\item $\mathcal{A}$, who has still access to the encryption oracle, produces a bit and wins the game if the bit is equal to $b$.
\end{enumerate}
\end{definition}

\section{Authentication}
\label{qmac}
Fortunately, defining quantum security for authentication is more straightforward.
\begin{definition}
\label{defqmac}
A MAC is \emph{existentially unforgeable under a quantum chosen message attack (EUF-qCMA)} if no adversary can win the following game, except with negligible probability.\\
\emph{Quantum MAC-forge game:}
\begin{enumerate}
\item  $\Gen(1^n)$ generates a random key.
\item $\mathcal{A}$ is given input $1^n$ and \emph{quantum} oracle access to $\Mac_k$. $\mathcal{A}$ wins if she outputs $q+1$ distinct pairs $(m,t)$ such that $\Vrfy_k(m,t)=1$, where $q$ is the number of oracle queries $\mathcal{A}$ made.
\end{enumerate}
The oracle uses the same randomness to generate the tags in one query. This is more conservative than assuming that new randomness is used for each tag and it is easier for the authenticator.
\end{definition}

\newpage
\newcommand{\mstar}{[m]^*}

\chapter{Quantum-Secure Authentication}
\label{main}
In this chapter we re-prove a theorem  first proved by Boneh and Zhandry and describe some implications of it.
\section{Main Result}
\label{6.1}
When we model the situation in which an adversary with quantum power is trying to forge a classical message authentication tag, we intuitively assume that the communication from the authenticating device is always classical.  However, as explained in Chapter \ref{qsecurity}, making this assumption is not the most conservative approach. It may be the case that the verifier or authenticating agents do have quantum computing power and measure their states just before communication. If we now consider an adversary who can somehow get access to the final quantum state, then 
we must allow her to make queries in superposition. 

Boneh and Zhandry~\cite{BZ13a} prove that it is possible to construct a MAC that is secure under \emph{quantum-chosen-message attacks} using a quantum-secure PRF. They develop a new technique for this purpose, which they call \emph{the rank method}. Using this method, they bound the probability that an adversary outputs  a number of input-output pairs of the oracle function which is bigger than the number of queries. 
A special case of the result is equivalent to a result by Farhi et al.~\cite{Far99}. In this case the oracle is a Boolean function and an upper bound is shown on the number of functions that can be distinguished by an algorithm making $q$ quantum queries. This upper bound is a function of the number of queries and the success probability $p$.

In the next section, we generalize the result and proof technique of the paper by Farhi et al.\ and thereby we re-prove Boneh and Zhandry's Theorem 4.1~\cite{BZ13a}. 
We believe that this new approach contributes to the insight into the strength of adversaries who are able to use both quantum computation and quantum communication.

Recall  the formal definition of a quantum secure MAC (Definition \ref{defqmac})
with a game between challenger and adversary in which the adversary is allowed to make $q$ quantum queries to the authentication oracle and her task is to output $k>q$ distinct message-tag  pairs. Theorem \ref{thm1} bounds the probability that an adversary wins this game if the tag function is chosen uniformly at random. To build a real secure MAC, this random function needs to be replaced by some pseudorandom function. Apart from this basic application, Boneh and Zhandry construct a quantum secure  variant of the Carter-Wegman MAC.

\begin{theorem}
\label{thm1}
Let $\mathcal{A}$ be a quantum algorithm making $q$ queries to a uniformly chosen oracle $f\colon X \rightarrow Y$ where $|X|= n,\ |Y|={\M}$. The probability that $\mathcal{A}$ successfully outputs $k>q$ distinct pairs $(x\in X,f(x)\in Y)$ is at most $\frac{1}{{\M}^k} \sum\limits_{i=0}^{q}  {k\choose i} ({\M}-1)^i$. 
\end{theorem}
To get a feeling for the size of the expression in Theorem \ref{thm1}, we note that if the sum would run from 0 to $k$, then it was equal to the binomial expression of $(m-1+1)^k$. Since $m>1$, every term is strictly positive; in particular the term $i=k$.
The factor $\frac{1}{m^k}$ ensures that the expression in the theorem is at most one if $q$ is smaller or equal to $k$, so it can be a probability. Because $q<k$ is required, the bound in Theorem \ref{thm1} is non-trivial (i.e.\ strictly smaller than 1).

We assume that the final state of the algorithm contains, when measured, a sequence $\xv$ of $k$ distinct elements of $X$ followed by a sequence of $k$ elements of $Y$. 
The rest of the register is called $w$ for working space. We write the final state of the algorithm as follows:
$$|\psi^f\rangle = \sum_z \alpha_z(f) |z\rangle = \sum_{x,y,w} \alpha_{xyw} |xyw\rangle$$
The success probability of the algorithm is the expectation (over the uniform distribution of oracles) of the probability that the first $k\cdot \log n$ qubits represent $k$ different $x$'s and the subsequent $k\cdot \log m$ qubits represent the sequence $f(\xv) = f(x_1)f(x_2)\cdots f(x_k)$ where $f$ is consistent with the particular oracle.
$$
\Pr[success] = \mathbb{E}_f\left[ \sum_{\xv\subseteq X, w} |\alpha_{\xv f(\xv)w}(f)|^2\right]$$
Like Boneh and Zhandry do, we assume that the sequence $\xv \subseteq X$ that the algorithm outputs does not depend on the oracle. After finishing the proof we argue that we can reduce the general case to the one with this assumption.\\

\noindent
The proof by Farhi et al.\ we will build upon makes use of the fact, proved by Beals et al.~\cite{Bea97}, that the amplitudes of the final states of any quantum algorithm are $2k$-degree multilinear polynomials in the values of the oracle-function, where $k$ is the number of queries. This is known as an essential part of \emph{the polynomial method} and is normally applied with Boolean inputs. The situation we model involves functions  with a larger codomain. We cannot represent our assumptions if we replace a query by multiple bitwise queries on a value since this will allow the adversary to query an arbitrary mixture of bits of function values. Therefore, we generalize the theorem by Beals et al.\ for functions with an arbitrary codomain $Y$. 
Proposition \ref{querytypes} shows that a superposition query implemented as a \emph{phase-changing} operation is equivalent to the commonly used one which adds the answer to the second register, so we can use the first implementation without changing the model. \emph{Oracle query}:
$$
|x,b,w \rangle \mapsto \omega_{\M}^{b \cdot f(x)} |x,b,w \rangle
$$
Here $\omega_{\M}$ is the ${\M}$th complex root of unity $e^{2\pi \frac{i}{{\M}}}$.

\section{Proof of Theorem \ref{thm1}}
With the following proposition, we mould the amplitudes of a state after $q$ queries into a nice polynomial form.
\begin{proposition}
\label{poly}
The amplitudes $\alpha_z$ in the state $\sum_z \alpha_z |z\rangle $ of any algorithm after making $q$ quantum queries to a function $f\colon X \rightarrow Y$ with $|Y|={\M}$, can be written as $$\sum_{\substack{S\subseteq X,  \\ \mathbf{b}\colon S \rightarrow \mstar}} \beta_{S,\mathbf{b}} \cdot \prod_{x \in S} \omega_{\M}^{f(x)\cdot \mathbf{b}(x)}$$ where $|S|\leq q$ and the function $\mathbf{b}$ is represented as set of input-output pairs and $[m]^* = \{1,2,\dots,m-1\}$.
\end{proposition}
\begin{proof} 

We prove this lemma by induction on the number of queries made by the algorithm. \\
\emph{Base}: if no queries have been made the sum in the lemma is taken over only one pair $\{S, \mathbf{b}\}$, namely $\{\emptyset, \emptyset \}$. The single term in this sum is a constant times an empty product so just the constant (as an empty product is 1). Amplitudes do not depend on the oracle if no queries are made so they are indeed constant. The base case holds.\\
\emph{Induction Hypothesis}: assume that the lemma holds for $q=t$; then the amplitudes of any algorithm after making $t$ quantum queries (each followed by a unitary) to a function $f$ can be written as $$\sum_{\substack{S\subseteq X,\\ \mathbf{b}\colon S \rightarrow  \mstar}} \beta_{S,\mathbf{b}} \cdot \prod_{x \in S} \omega_{\M}^{f(x)\cdot \mathbf{b}(x)}$$ where $|S|\leq t$. We compute the state after one more query.\\ 
\emph{Step}:
the basis state $|x_0b_0w_0 \rangle$ has amplitude $\sum_{S,\mathbf{b}} \beta_{S,\mathbf{b}} \cdot \prod_{x \in S} \omega_{\M}^{f(x)\cdot \mathbf{b}(x)}$ after the unitary transform that follows the $t^{th}$ query. The next query only changes the global phase of a each state: each amplitude $\alpha_{x_0b_0w_0} \text{ becomes }  \alpha_{x_0b_0w_0}\cdot \omega^{f(x_0)\cdot b_0}$. We express the new amplitude in terms of the old:
$$\alpha_{x_0b_0w_0} = 
\left( \sum_{\substack{S\subseteq X,  \\ \mathbf{b}\colon S \rightarrow  \mstar}} \beta_{S,\mathbf{b}} \cdot \prod_{x \in S} \omega_{\M}^{f(x)\cdot \mathbf{b}(x)} \right)
\cdot
\omega^{f(x_0)\cdot b_0}$$ Here $|S|\leq t$.\\
Each term of the sum is multiplied by $\omega^{f(x_0)\cdot b_0}$ which transforms the sum in one of the following ways:
\begin{itemize}
\item If $x_0 \not \in S$ then the product is now over $x\in S\cup \{x_0\}$. For the sum it means that the term $(S,\textbf{b})$ disappears and a new term $(S\cup\{x_0\}, \mathbf{b}\cup \{(x_0,b_0)\}$ with the same coefficient $\beta_{S, \mathbf{b}}$ is added. Note that $|S|$ now can be $t+1$.
\item If $x_0 \in S$ and $b(x_0) + b_0 \neq 0 $ then the term $(S, \mathbf{b})$ is replaced by the term $(S, \mathbf{b}')$ where $\mathbf{b}'(x_0) = [\mathbf{b}(x_0)+b_0 \mod {\M}]$. Again, the new term gets the coefficient of the disappearing term.
\item If $x_0 \in S$ and $\mathbf{b}(x_0) + b_0 \equiv 0 \pmod {\M}$ then this term is replaced by a term $(S\backslash \{x_0\}, \mathbf{b} \backslash \{(x_0,\mathbf{b}(x_0)\})$.

\end{itemize}
If we added or changed more than one term with the same pair $\{S, \mathbf{b} \}$, we sum the coefficients to make it a single term. By the induction hypothesis, each $\mathbf{b}(x) \text{ is in } \mstar$ before the step. We added these numbers modulo ${\M}$, and when the sum was equivalent to 0, the $x$ disappeared from the set $S$. This means that the new $\mathbf{b}(x)$ are all in $\mstar$.

The amplitudes are of the same form as before the last query, now with $|S|\leq t+1$. The number of terms in the sum did not increase. The next part of the induction step is the application of some unitary. After multiplying the amplitude vector by a unitary matrix, each amplitude is a linear combination of the amplitudes after the query. This linear combination can be written in the same form, with new coefficients. 
Now, the number of terms of each sum can be much bigger: at most the number of pairs of a set $S$ of size at most $t+1$ together with a function $\mathbf{b}$. 
The number of distinct pairs $(S, \mathbf{b})$ is $\sum_{i=0}^{k}  {|X|\choose i} ({\M}-1)^i$. For each size $i$ of set $S$ there are ${|X| \choose i}$ possible sets, each of which can come with $({\M}-1)^i$ possible functions $\mathbf{b}$. Here we see a glimpse of the expression of Theorem \ref{6.1}.

The induction step shows that if the amplitudes resulting from any algorithm making $t$ queries have the claimed form with $|S|\leq t$, then an algorithm making $t+1$ queries has amplitudes of the same form with $|S|\leq t+1$. By induction, the final state of any algorithm making $q$ queries has amplitudes of the claimed form with $|S|\leq q$. \qed

\end{proof}
Now we can start to bound the success probability of a MAC forger. A couple of lemmas lead to Theorem \ref{thm1}.
\begin{lemma}
\label{beta}
Let $f\colon X \rightarrow Y$ be a function, $|X|=n,\ |Y|=m$ and let $\alpha: Y^X \rightarrow \mathbb{C}$ be a multivariate function of the form $$\alpha(f) = \sum\limits_{\substack{S\subseteq X,  \\ \mathbf{b} \colon S \rightarrow \mstar}} \beta_{S,\mathbf{b}} \cdot \prod_{x \in S} \omega_{\M}^{f(x)\cdot \mathbf{b}(x)}$$ 
where $|S|\leq q$.\\
If $\exists f$ such that $\alpha(f) =1$ then the following holds:
$$\sum_{S\subseteq X,  \ \mathbf{b} \colon S \rightarrow \mstar} |\beta_{S,\mathbf{b}}|^2 \geq \frac{1}{\sum\limits_{i=0}^{q}  {n\choose i} (m-1)^i}\ .$$
\end{lemma}
\begin{proof}
We assume that $\alpha$ is of the stated form and $\alpha(\hat{f}) =1$.
\begin{align*}
\alpha({\hat{f}}) = & \sum_{S\subseteq X,  \ \bfn \colon S \rightarrow \mstar} \beta_{S,\bfn} \cdot \prod_{x \in S} \omega_{\M}^{\hat{f}(x)\cdot \bfn(x)} =1 \\
\end{align*}
By the Cauchy-Schwarz inequality we have:
$$|1|^2= \mid \sum_{S, \bfn} \beta_{S,\bfn} \cdot \prod_{x \in S} \omega_{\M}^{\hat{f}(x)\cdot \bfn(x)}\mid^2 \leq 
\sum_{S, \bfn} \mid \beta_{S,\bfn} \mid^2 \cdot \sum_{S, \bfn} |\prod_{x \in S} \omega_{\M}^{\hat{f}(x)\cdot \bfn(x)}|^2 $$
Since the norm of any power of omega is 1, this gives:
$$1= | \sum_{S, \bfn} \beta_{S,\bfn} \cdot \prod_{x \in S} \omega_{\M}^{\hat{f}(x)\cdot \bfn(x)}|^2 \leq 
\sum_{S, \bfn} |\beta_{S,\bfn}|^2 \cdot \sum_{S, \bfn} 1$$
We divide by the sum over ones (which is equal to the number of terms of the sum):
$$\sum_{S, \bfn} |\beta_{S,\bfn}|^2 \geq \frac{1}{\# (S,\bfn)}$$
The number of terms in the sum, $\# (S,\bfn)$ can be counted by summing over all sizes of $S$. For each size $i$ there are $n \choose i$ subsets of $X$ and for each subset there are $(m-1)^i$ functions $\bfn$ from $S$ to $\mstar$.  $$\# (S,\bfn) = \sum_{i=0}^{q} {n \choose i} (m-1)^i$$
This proves the lemma.  \qed
\end{proof}

\begin{lemma}
\label{last}
Let 
$\alpha(g)$ a function of the form: $$\sum\limits_{S\subseteq X,  \ \bfn \colon S \rightarrow \mstar} \beta_{S,\bfn} \cdot \prod\limits_{x \in S} \omega_{\M}^{g(x)\cdot \bfn(x)}$$
for $g: X\rightarrow Y$. Then:
$$\mathbb{E}_g \left[|\alpha(g)|^2\right] = \sum_{S,\bfn}\left|\beta_{S,\bfn}\right|^2$$
where $g$ is chosen uniformly from the set of all functions from $X$ to $Y$.
\end{lemma}
\begin{proof}
\begin{align*}
\mathbb{E}_g \left[|\alpha(g)|^2\right]& = \mathbb{E}_g \left[|\sum_{S,\bfn} \beta_{S\bfn} \cdot \prod_{x\in S} \omega^{g(x)\bfn(x)}|^2 \right]\\
& = \mathbb{E}_g\left[ \sum_{(S_1,\bfn_1),(S_2,\bfn_2)} \beta_{S_1\bfn_1} \overline{\beta_{S_2\bfn_2}} \prod_{x\in S_1} \omega^{g(x)\bfn_1(x)} \overline{\prod_{x'\in S_2} \omega^{g(x')\bfn_2(x')}}\right]
\end{align*}
Since $\beta_{S\bfn}$ is independent from $g$ we can move the expectation  in front of the products after swapping it with the sum which we may do by linearity.
\begin{equation}
\label{expect}
\mathbb{E}_g \left[|\alpha(g)|^2\right] = \sum_{\substack{(S_1,\bfn_1)\\(S_2,\bfn_2)}} \beta_{S_1\bfn_1} \overline{\beta_{S_2\bfn_2}} \  \mathbb{E}_g\left[ \prod_{x\in S_1} \omega^{g(x)\bfn_1(x)} \overline{\prod_{x'\in S_2} \omega^{g(x')\bfn_2(x')}}\right]
\end{equation}
We distinguish three different kinds of terms in this sum: terms with $S_1 \neq S_2$, terms with $S_1 = S_2$ and $\bfn_1 \neq \bfn_2$, and terms with $(S_1,\bfn_1)=(S_2,\bfn_2)$. Terms of the first two kinds can be partitioned in groups that sum to 0, which leaves only the terms of the last type.

\begin{itemize}
\item case $S_1\neq S_2$\\
Without loss of generality we can assume that $\exists x_d$ such that $x_d \in S_1,\ x_d\not \in S_2$. Now we fix values $g(x)$ for all $x\neq x_d$. Let $G_d$ be the set of functions that are consistent with the fixed values. In this set there is exactly one function for each possible value $g(x_d)$. 
\begin{align*}
\mathbb{E}_{G_d} \left[ \prod_{x\in S_1} \omega^{g(x)\bfn_1(x)}  \overline{\prod_{x\in S_2} \omega^{g(x)\bfn_2(x)}} \right]& = \\
\mathbb{E}_{G_d} \left[ \omega^{g(x_d)\bfn_1(x_d)}\prod_{x\neq x_d} \omega^{g(x)\bfn_1(x)}  \overline{\prod_{x\in S_2} \omega^{g(x)\bfn_2(x)}} \right]&
\end{align*}
The powers of $\omega$ in the product are the same for all functions in $G_d$, so the products do not depend on $g$. We can view the big products as a constant $a$ and take it out of the expectation.
$$a \cdot \mathbb{E}_{g\in G_d} \left[\omega^{g(x_d)\bfn(x_d)}\right]
= a \cdot \mathbb{E}_{y\in Y} \left[ \omega_m^{y\bfn(x_d)}\right]$$
By Lemma \ref{zerosum} and because $Y = [m]$, this is 0.\\
We can partition all functions in $Y^X$ in parts with $g(x)$ fixed for $x\neq x_d$. Each part has expected value 0 for the product so the expected value over all functions is also 0.

\item case $S_1=S_2$ and $\bfn_1\neq \bfn_2$\\
$\exists x_d$ such that $\bfn_1(x_d)\neq \bfn_2(x_d)$.
Again we partition all functions $g$ such that the functions in one part have the same values on $x\neq x_d$.

\begin{align*}
\mathbb{E}_{G_d} \left[\prod_{x\in S_1} \omega^{g(x)\bfn_1(x)} \overline{\prod_{x\in S_2} \omega^{g(x)\bfn_2(x)}}\right] \\=
\mathbb{E}_{G_d} \left[\omega^{g(x_d)\bfn_1(x_d)} \overline{\omega^{g(x_d)\bfn_2(x_d)}} \prod_{x\neq x_d} \omega^{g(x)\bfn_1(x)} \overline{\prod_{x\neq x_d} \omega^{g(x)\bfn_2(x)}}\right]
\\=
\mathbb{E}_{G_d} \left[\omega^{g(x_d)(\bfn_1(x_d)-\bfn_2(x_d))}  \prod_{x\neq x_d} \omega^{g(x)\bfn_1(x)} \overline{\prod_{x\neq x_d} \omega^{g(x)\bfn_2(x)}}\right]
\end{align*}
In this equation the big products, as well as $(\bfn_1(x_d)- \bfn_2(x_d))$, are constants, so again the expected value is 0 by Lemma \ref{zerosum}.

\item case $(S_1,\bfn_1)=(S_2,\bfn_2)$
\begin{align*}
\mathbb{E}_{g} \left[\prod_{x\in S_1} \omega^{g(x)\bfn_1(x)} \overline{\prod_{x\in S_2} \omega^{g(x)\bfn_2(x)}}\right] \\=
\mathbb{E}_{g} \left[\prod_{x\in S_1} \omega^{g(x)\bfn_1(x)} \overline{\prod_{x\in S_1} \omega^{g(x)\bfn_1(x)}}\right] \\=
\mathbb{E}_{g} [1]
\end{align*}
\end{itemize}
In the sum of equation \ref{expect}, 
only the terms remain with $(S_1,\bfn_1)=(S_2,\bfn_2)$ and these terms consist of the two coefficients multiplied by 1:
$$\mathbb{E}_g [|\alpha(g)|^2] = \sum_{(S_1,\bfn_1)=(S_2,\bfn_2)} \beta_{S_1\bfn_1} \overline{\beta_{S_2\bfn_2}} \cdot 1=
\sum_{(S,\bfn)} |\beta_{S\bfn}|^2 
$$

\qed
\end{proof}
As explained in Section \ref{6.1}, the final state of any quantum algorithm making $q$ oracle queries to $f \in Y^X$ can be written as 
$$|\psi^f\rangle = \sum_z \alpha_z(f) |z\rangle = \sum_{x,y,w} \alpha_{xyw} |xyw\rangle$$
and the probability that the adversary succeeds is
\begin{equation}
\label{defp}
\Pr[success] = \mathbb{E}_f\left[\sum_{\substack{\xv\subseteq X,\\ w}} |\alpha_{\xv f(\xv)w}(f)|^2\right]=p.
\end{equation}
Here, $\xv$ has the required length and $f$ is uniformly chosen from $Y^X$.
Since it is a quantum  state, the squares of the coefficients $\alpha_{xyw}$ add up to (at most) one, for a fixed $f$.\\

\noindent
Fix $\xv \subset X$, $w$ and $f$. Having this fixed $\xv$ we can split each function $g = g_0\cup g_1$ where $g_0$ has domain $\xv$ and $g_1$ acts on the complement $X\backslash \xv$. We now define
\begin{equation}
\label{alphaprime} 
\alpha'_{\xv f w}(g_0):= \frac{\alpha_{\xv f(\xv) w}(g_0\cup f_1)}{\alpha_{\xv f(\xv) w}(f)}
\end{equation} 
which satisfies $\alpha'_{\xv f w}(f_0)=1$.

\begin{corollary}
\label{cor}
For any fixed $\xv \subset X$ of size $k$, $w$ and $f$ the following holds:
$$\mathbb{E}_g [|\alpha'_{\xv fw}(g_0)|^2]\geq \frac{1}{\sum\limits_{i=0}^{q} {k \choose i} ({\M}-1)^i}$$ 
\end{corollary}
\begin{proof}
We will show that we can apply the previous two lemmas and that $|X|$ in the binomial coefficient becomes $k$. Firstly, we write $\alpha'$ as we defined it (\ref{alphaprime}).
$$\mathbb{E}_g [|\alpha'_{\xv fw}(g_0)|^2]= \mathbb{E}_g [|\frac{\alpha_{\xv f(\xv)w}(g_0\cup f_1)}{\alpha_{\xv f(\xv)w}(f)}|^2]$$
By definition, $\alpha'$ does not depend on $g_1$ so we may as well take the expectation according to the uniform distribution over $g_0$ instead of $g$. 
$$\mathbb{E}_{g_0} [|\alpha'_{\xv fw}(g_0)|^2]= \mathbb{E}_{g_0} [|\frac{\alpha_{\xv f(\xv)w}(g_0\cup f_1)}{\alpha_{\xv f(\xv)w}(f)}|^2]$$
We have shown  (Proposition \ref{poly}) that the numerator is of the polynomial form required to use Lemma \ref{last}. The denominator is a constant since we fixed $\xv$, $w$ and $f$. Each coefficient of $\alpha_{\xv f(\xv)w}(g_0\cup f_1)$ is divided by a constant, which means that the form is preserved. So $\alpha'_{\xv fw}(g_0)$ is of the form $\sum\limits_{S\subseteq X,  \ b\colon S \rightarrow \mstar} \beta'_{S,b} \cdot \prod\limits_{x \in S} \omega_{\M}^{g_0\cup f_1(x)\cdot b(x)}$ where $\beta'_{Sb} = \frac{\beta_{Sb}}{\alpha_{\xv f(\xv)w}(f)}$
and $|S|\leq q$.\\ 
We can rewrite this polynomial a bit since $g_0\cup f_1 (x) = f_1(x) $ is constant if $x\not \in \xv$. Each constant power of $\omega$ can be moved to the coefficient which leaves us with new coefficients $\beta''_{Sb}$. The products are over $x\in (S\cap \xv)$ so we can now combine the terms that multiply over the same intersection. This results again in new coefficients $\beta'''_{S\cap\xv ,b}$.\\ 
Now by Lemma \ref{last}: $\mathbb{E}_g [|\alpha'(g_0)|^2] = \sum\limits_{S \subseteq \xv,\ b \colon \xv \rightarrow \mstar }|\beta'''_{S,b}|^2$ where still $|S|\leq q$.\\
Because $\exists f$, namely our fixed $f$, such that $\alpha'(f)=1$, by Lemma \ref{beta} we have
\begin{align*}
\sum_{S\subseteq \xv,  \ b\colon S \rightarrow \mstar} |\beta'''_{S,b}|^2 & \geq \frac{1}{\sum\limits_{i=0}^{q}  {|\xv|\choose i} ({\M}-1)^i}
\intertext{So} \mathbb{E}_{g_0} [|\alpha'(g_0)|^2] & \geq \frac{1}{\sum\limits_{i=0}^{q}  {k\choose i} ({\M}-1)^i} 
\end{align*}
\qed
\end{proof}
To prove Theorem \ref{thm1}, we use our definition of $\alpha'$ in the inequality of Corollary \ref{cor}.
\begin{align*}
\mathbb{E}_{g_0} [|\frac{\alpha_{\xv f(\xv)w(g_0\cup f_1)}}{\alpha_{\xv f(\xv)w}(f)}|^2]  \geq \frac{1}{\sum\limits_{i=0}^{q}  {k\choose i} ({\M}-1)^i}
\end{align*}
The oracle function the adversary has access to is chosen uniformly at random from $Y^X$, and the restriction of $g$ to $\xv$ is uniformly chosen from $Y^{\xv}$. Therefore we can replace the expectation by a summation and divide by the number of functions $|Y|^{|\xv|}={\M}^k$.
\begin{align*}
\sum_{g_0\in Y^{\xv}} [|\frac{\alpha_{\xv f(\xv)w(g_0\cup f_1)}}{\alpha_{\xv f(\xv)w}(f)}|^2]  \geq \frac{{\M}^{k}}{\sum\limits_{i=0}^{q}  {k\choose i} ({\M}-1)^i}
\end{align*}
The constant $\alpha_{\xv f(\xv)w}(f)$ can be taken out of sum after the squared norm   is taken of the numerator and the denominator separately. If we then take the denominator to the right hand side  we get:
\begin{align*}
\sum_{g_0\in Y^{\xv}} |\alpha_{xf(\xv)w}(g_0 \cup f_1)|^2 & \geq \frac{{\M}^{k} \cdot |\alpha_{\xv f_0w}(f)|^2}{\sum\limits_{i=0}^{q}  {k\choose i} ({\M}-1)^i}
\intertext{Up to this point we worked with fixed $\xv$, $f$ and $w$. Now we sum over all possible choices.}
\sum_f \sum_{\xv} \sum_w \sum_{g_0} |\alpha_{\xv f(\xv)w}(g_0 \cup f_1)|^2 & \geq 
\frac{{\M}^{k} \cdot \sum_f \sum_\xv \sum_w |\alpha_{\xv f(\xv)w}(f)|^2}{\sum\limits_{i=0}^{q}  {k\choose i} ({\M}-1)^i}
\intertext{On the left-hand side we pull the sum over $g$ outside the new sums, and we split the sum over $f$ into sums over $f_0$ and $f_1$. }
\sum_{g_0} \sum_{f_1} \sum_{w, \xv, f(\xv)}  |\alpha_{\xv f(\xv)w}(g_0 \cup f_1)|^2 & \geq 
\frac{{\M}^{k} \cdot \sum_f \sum_\xv \sum_w |\alpha_{\xv f(\xv)w}(f)|^2}{\sum\limits_{i=0}^{q}  {k\choose i} ({\M}-1)^i}
\end{align*}
{The sum over $w$, $\xv$ and $f(\xv)$ of the squared amplitudes with a fixed oracle function, is at most 1 since the amplitudes are in the final state of $\mathcal{A}^{g_0\cup f_1}$.
The sum over all possible oracles $g_0\cup f_1$ is at most $|Y^X|={\M}^{n}$.\\
On the right-hand side we replace the sum over $f$ by the expectation over the uniform distribution of $f$ multiplied by the number of functions.}

\[
{\M}^{n} \geq 
\frac{{\M}^{k} \cdot {\M}^{n} \cdot \mathbb{E}_f[ \sum\limits_{\xv} \sum\limits_w |\alpha_{\xv f(\xv)w}(f)|^2]}{\sum\limits_{i=0}^{q}  {k\choose i} ({\M}-1)^i}
\]
The expectation over $f$ is exactly $p$, the success probability of the adversary (\ref{defp}).
$$1 \geq \frac{{\M}^k\cdot p}{\sum\limits_{i=0}^{q}  {k\choose i} ({\M}-1)^i}$$
This is equivalent to:
$$p \leq \frac{1}{{\M}^k} \sum\limits_{i=0}^{q}  {k\choose i} ({\M}-1)^i$$
This finishes the proof for the special case in which the sequence $\xv$, consisting of the input parts of the input-output pairs (generated by the adversary in the game) is independent of the oracle. In other words, the case in which the adversary decides in advance 
for which inputs she will give the function values.

To cover the general case, we now assume that there exists an adversary $\mathcal{A}$ violating Theorem \ref{thm1}, who outputs a sequence $\xv$ that depends on the oracle $f$. We use this adversary to construct a new algorithm $\mathcal{B}$ that successfully outputs $k>q$ distinct pairs $(x,f(x))$ with the same probability as $\mathcal{A}$ but with the difference that the sequence $\xv$ is independent of the oracle.
$\mathcal{B}$ picks a random oracle $O$ with the same domain $X$ and range $Y=[m]$ as $f$. It simulates $\mathcal{A}$ giving it the oracle $f+O$ defined by $f+O(x) := f(x) + O(x) \mod m$. 
$\mathcal{A}$ sees a random oracle and succeeds with probability greater than the bound from Theorem \ref{thm1}. The sequence $\xv$ of $\mathcal{A}$'s output may depend on $f+O$ but is independent of $f$ and $O$ which are both independent of their sum and not seen by $\mathcal{A}$. $\mathcal{B}$ can translate the pairs given by $\mathcal{A}$ to input-output pairs of the function $f$ by subtracting $O(x)$ from each $y$. The probability that $\mathcal{A}$'s pairs are distinct and correct is equal to the probability that $\mathcal{B}$'s pairs are distinct and correct. 
With a sequence $\xv$ as output that is independent from its oracle, $\mathcal{B}$ belongs to the case for which we proved Theorem \ref{thm1}. 
It follows that the success probability of $\mathcal{B}$ is bounded by Theorem \ref{thm1} and therefore the same is true for $\mathcal{A}$'s success probability. Our assumption is false and the theorem holds in the general case.
\qed

\section{Applications of Theorem \ref{thm1}}
The bound in Theorem \ref{thm1} is a complex term and it may be difficult to feel what this bound tells us. Towards the application in the security game of a message authentication code, we consider the case in which $k=q+1$. We then have a bound on the success probability of 
$$\frac{1}{{\M}^{q+1}} \sum\limits_{i=0}^{q}  {q+1\choose i} ({\M}-1)^i
=\frac{1}{{\M}^k} \left(\sum\limits_{i=0}^{q+1}  {q+1\choose i} ({\M}-1)^i -(m-1)^{q+1}\right)
$$
Using Newton's binomial theorem we can replace the sum in this expression by $(1+(m-1))^{q+1} = m^{q+1}$.
The bound becomes the following:
$$\frac{1}{{\M}^{q+1}} (m^{q+1}-(m-1)^{q+1}) = 1-\frac{(m-1)^{q+1}}{m^{q+1}} = 1- (1-\frac{1}{m})^{q+1}$$
By induction on $q+1$, $(1-\frac{1}{m})^{q+1} \geq 1- \frac{q+1}{m}$. Base: $(1-\frac{1}{m})^{1} \geq 1- \frac{1}{m}$. 
Induction step: assume $(1-\frac{1}{m})^{n} \geq 1- \frac{n}{m}$, then $(1-\frac{1}{m})^{n+1} \geq (1- \frac{n}{m})(1-\frac{1}{m}) = 1-\frac{n}{m}-\frac{1}{m}+\frac{n}{m^2}\geq 1- \frac{n+1}{m}$.\\
The bound of Theorem \ref{6.1} is thus at most  $ \frac{q+1}{n}$.\\

\noindent
An immediate consequence of Theorem \ref{thm1} is the following: a quantum secure pseudorandom function can serve as a MAC that is secure against superposition attacks. 
The existence of quantum pseudorandom functions, that means PRFs that are indistinguishable from a random function by adversaries making superposition queries (assuming that classical pseudorandom functions exist), is proved by Zhandry \cite{Zha12b}.
\begin{theorem}
Let $f: K \times X \rightarrow Y$ be a quantum pseudorandom function. If $\frac{1}{|Y|}$ is negligible in $n$  then $f_k(m)$ is a EUF-qCMA-secure MAC as defined in Section \ref{qmac}.
\end{theorem}
\begin{proof}
Any adversary $\mathcal{A}$ making polynomially many queries ($q$) to the MAC oracle can win the EUF-qCMA game only with probability $\frac{q}{|Y|}$ if the oracle is a random function. This probability is negligible in the length of the tags $\log(|Y|)$ and the length of the messages.
$f_k$ is indistinguishable from a random function so the difference between the success probability in the game with the random function as MAC and that of the game with $f_k$ as MAC is negligible. This means that $\mathcal{A}$ can only win the game in which the MAC oracle is $f_k$ with negligible probability.
\qed
\end{proof}
Boneh and Zhandry \cite{BZ13a} show that a modification of the Carter-Wegman MAC is secure against superposition attacks. Next, they consider one-time MACs and more general $q$-time MACs. They show that unlike classically, a ($q$+1)-wise independent function is not enough to build a secure $q$-time MAC, but a ($c+2q$)-wise independent function does ensure security. 

In their most recent paper \cite{BZ13b} they go further with giving a construction to turn any classically secure MAC into a quantum secure one by combining it with a \emph{chameleon hash function}. They also use a generalization of Theorem \ref{thm1} to prove that the GPV signature scheme built from pre-image sampleable functions and PRFs is quantum secure.

\newpage				
\chapter{Conclusion and Further Research}
\label{conclusion}
Not all traditional security definitions have a trivial counterpart in the quantum world. It is therefore interesting and insightful to study several options for modelling reality. It is not always needed nor practical to obtain the most conservative security notion covering the broadest class of attacks, but it is important to know whether reasonable-looking definitions are feasible at all. 

Comparing the new definition we proposed in this thesis (Definition \ref{newdef}) to other options is interesting for further research. It is not unlikely that our definition can be proved 
equivalent to the definition Boneh and Zhandry  suggest. 
Another question is how the real world relates to these definitions. The current state of the development of physical quantum computers can be used to make more specific assumptions concerning quantum adversaries.  

Our alternative proof of Boneh and Zhandry's theorem \cite{BZ13a} helps in understanding the (limitations on) advantages an adversary has from getting information in superposition. The new proof exposes a relation between research questions in post-quantum cryptography and existing results from the field of quantum computing. Specifically, we generalized and followed the lines of the proof by Farhi et al.\ \cite{Far99} of the theorem that states the following: any adversary making $q$ quantum queries to an oracle holding one of a set of Boolean functions cannot identify the function with probability at least $p$ if the set is of size $ > \frac{1}{p} [1+ {n \choose 1} + {n \choose 2} + \cdots + {n \choose q} ] $, where $n$ is the size of the domain of the functions. In our case, the oracle may hold any function (not only Boolean) and the goal of the adversary is not necessarily to know exactly which function it is, but to output $k$ input-output pairs. In any non-trivial case we require $k>q$.

The quantum-secure MACs we discussed in Chapter \ref{main} are an applications of \emph{quantum pseudorandom functions}, as defined and constructed by Zhandry~\cite{Zha12b}. 
Another application could be constructing \emph{quantum pseudorandom permutations}. 
Classical pseudorandom permutations can be built from pseudorandom functions using for example a Feistel network. These techniques are widely used as block ciphers such as DES. Whether the classical Luby-Rackoff construction \cite{LR88} can be shown secure against quantum attacks 
is an interesting open question. Naor and Reingold \cite{NR99} gave a clear proof of the correctness of the classical construction, which can be a starting point for further research on quantum-secure cryptography.

With still some time before our real adversaries build or buy powerful quantum computers, we already have various results that contribute to the confidence in classical cryptography in a quantum world.  
The new paradigm in which there is quantum communication between the honest parties and the adversary has created new research questions. The quantum-secure MACs we discussed is an example, and there are other related questions to study.

\bibliography{thesisbMVelema}{}
\bibliographystyle{amsalpha}
\end{document}